\pdfoutput=1
\documentclass[sigconf]{acmart}

\renewcommand\footnotetextcopyrightpermission[1]{} % removes footnote with conference information in first column

\usepackage[utf8]{inputenc}
\usepackage{graphicx}
\usepackage{listings}
\usepackage{comment}
\usepackage[inline]{enumitem}
\usepackage[ruled,vlined,linesnumbered,noresetcount]{algorithm2e}
\usepackage{float}
\usepackage{tikz}
\usetikzlibrary{calc,arrows}
\tikzset{
  abstraction/.style={
    draw,
    rectangle,
    rounded corners,
    inner sep=.15cm,
    minimum
    height=.5cm,
    minimum width=1cm,
    align=left,
    anchor=north west,
  },
}
\usepackage{etoolbox}
\makeatletter
\patchcmd{\@algocf@start}% <cmd>
  {-1.5em}% <search>
  {0pt}% <replace>
  {}{}% <success><failure>
\makeatother

\newcommand{\SB}[1]{}
\newcommand{\AD}[1]{}
\newcommand{\ST}[1]{}
\newcommand{\AG}[1]{}

\newcommand{\ie}{i.e.}
\newcommand{\eg}{e.g.}
\newcommand{\etal}{\emph{et al.}}

\newcommand{\seq}[1]{\overline{#1}}
\newcommand{\seqAt}[2]{{#1}^{#2}}
\newcommand{\seqUpTo}[2]{\overline{#1^{#2}}}

\newcommand{\seqPr}[2]{\overline{#1}_{#2}}
\newcommand{\seqPrAt}[3]{{#1}_{#2}^{#3}}
\newcommand{\seqPrUpTo}[3]{\overline{#1^{#3}_{#2}}}

\newcommand{\fn}[1]{\ensuremath{\mathsf{#1}}}
\newcommand{\validTxs}{\fn{validLog}}
\newcommand{\validBlocks}{\fn{validChain}}
\newcommand{\fusion}{\fn{fusion}}

\newcommand{\submit}{\mathit{submit}}
\newcommand{\request}{\ensuremath{\mathit{request}}}
\newcommand{\inputBlock}{\mathit{input}}
\newcommand{\outputBlock}{\mathit{output}}
\newcommand{\readChain}{\mathit{read}}
\newcommand{\client}{\ensuremath{\mathit{client}}}
\newcommand{\replica}{\ensuremath{\mathit{replica}}}
\newcommand{\subProposal}{\mathit{subProposal}} % for pipelining, may be removed

\newcommand{\tx}{\ensuremath{\mathit{tx}}}

\newcommand{\RC}{\ensuremath{\mathrm{RC}}}
\newcommand{\FRC}{\ensuremath{\mathrm{FRC}}}
\newcommand{\pool}{\ensuremath{\mathit{pool}}}
\newcommand{\txs}{\ensuremath{\mathit{txs}}}
\newcommand{\done}{\ensuremath{\mathit{done}}}
\newcommand{\prop}{\ensuremath{\mathit{prop}}} % for pipelining, may be removed
\newcommand{\QC}{\ensuremath{\mathit{QC}}} % for pipelining, may be removed

\newcommand{\assign}{\leftarrow}
\newcommand{\append}{{}^\frown}
\newcommand{\concat}{\mathbin{::}}

\newcommand{\REQ}{\ensuremath{\texttt{REQ}}}
\newcommand{\SUBPROP}{\ensuremath{\texttt{SUB-PROP}}}

\SetKwBlock{SubAlgoBlock}{}{end}
\newcommand{\SubAlgo}[2]{#1 \SubAlgoBlock{#2}}
\SetKw{Procedure}{procedure}
\SetKw{Upon}{upon}

\definecolor{dimgray}{gray}{0.50}
\definecolor{myredcol}{rgb}{0.9,0.17,0.31}   % Amaranth
\definecolor{mybluecol}{rgb}{0.0,0.58,0.71}  % Bondi blue

\AtBeginDocument{%
  \providecommand\BibTeX{{%
    \normalfont B\kern-0.5em{\scshape i\kern-0.25em b}\kern-0.8em\TeX}}}

\setcopyright{none}
\copyrightyear{2018}
\acmYear{2018}
\acmDOI{10.1145/1122445.1122456}

\begin{document}

\title{SoK: Achieving State Machine Replication in Blockchains based on Repeated Consensus}

\author{Silvia Bonomi}
\email{bonomi@diag.uniroma1.it}
\orcid{0001-9928-5357}
\affiliation{%
  \institution{Sapienza University}
  \city{Rome}
  \country{Italy}
}

\author{Antonella Del Pozzo}
\email{antonella.delpozzo@cea.fr}
\orcid{0003-0913-2141}
\affiliation{%
  \institution{Université Paris-Saclay, CEA, List}
  \city{Palaiseau}
  \country{France}
  \postcode{F-91120}
}

\author{Álvaro García-Pérez}
\email{alvaro.garciaperez@cea.fr}
\orcid{0002-9558-6037}
\affiliation{%
  \institution{Université Paris-Saclay, CEA, List}
  \city{Palaiseau}
  \country{France}
  \postcode{F-91120}
}

\author{Sara Tucci-Piergiovanni}
\email{sara.tucci@cea.fr}
\orcid{0001-9738-9021}
\affiliation{%
  \institution{Université Paris-Saclay, CEA, List}
  \city{Palaiseau}
  \country{France}
  \postcode{F-91120}
}

% \renewcommand{\shortauthors}{Bonomi and Del Pozzo, et al.}

\begin{abstract}
  This paper revisits the ubiquitous problem of achieving state machine replication in
  blockchains based on repeated consensus, like Tendermint. To achieve state machine replication in  blockchains built on top of consensus, one needs to guarantee fairness of user transactions.
  A huge body of work has been carried out on the
  relation between state machine replication and consensus in the past years, in a variety
  of system models and with respect to varied problem specifications.  
  We systematize this work by proposing novel and rigorous abstractions for state machine replication and repeated consensus in a system model that
  accounts for realistic blockchains in which blocks may contain several transactions
  issued by one or more users, and where  validity and order of transactions within a block is determined by an external application-dependent
  function that can capture various approaches for order-fairness in the
  literature. \AG{Revise and tone down if needed.} Based on these abstractions, we propose a reduction from state machine replication to repeated consensus, such that user fairness is achieved using the consensus module as a black box. This approach allows to achieve fairness as an add-on on top of preexisting consensus modules in blockchains based on repeated consensus.   %Furthermore, we propose a generic, time-constant transformation that
  %optimises this reduction by using pipelining techiques. The optimised transformation is
  %not a reduction proper because it uses the consensus protocol underlying the blockchain
  %as a whitebox. However, we show that the optimised transformation is truly generic since
  %it can be applied to any consensus protocol that includes a voting phase for deciding on
  %a value (virtually all of them!), and it can therefore be implemented in the blockchain
  %setting in a seamless way. Finally, we show how to apply the optimised transformation to
  %popular blockchains such as Cosmos and Algorand, thus adding support for state machine
  %replication to the services they provide. 
  %Our solution stems from and systematises
  %previous works on this topic. The novelty of our contribution lies in the formalisation
 % of the problem in the blockchain setting in a rigorous and general way.
\end{abstract}

%%% Local Variables:
%%% fill-column: 90
%%% require-final-newline: t
%%% mode-require-final-newline: t
%%% next-line-add-newlines: nil
%%% show-trailing-whitespaces: t
%%% indent-tabs-mode: nil
%%% ispell-dictionary: "british"
%%% mode: latex
%%% TeX-PDF-mode: t
%%% TeX-master: "main"
%%% End:




\maketitle
\pagestyle{plain} % removes running headers

\section{Introduction}\label{sec:intro}

The rise of blockchains \cite{algorand,cosmos,Nakamoto2008,Wood14} in the last decade has
revived the interest in distributed programming and Byzantine fault-tolerant systems. A
blockchain \cite{Nakamoto2008} is an append-only, tamper-resistant distributed ledger of
transactions organised in a chain of blocks, originally used for
cryptocurrencies. Ethereum \cite{Wood14} popularised the concept of smart contracts as a
mean to implement any application on top of a blockchain. These first proposals, however,
have been criticised because the underlying blockchain protocol used to update the ledger,
based on proof-of-work, is extremely energy consuming. New proposals then subsequently
emerged proposing alternatives to proof-of-work, including a notable number using standard
Byzantine fault-tolerant distributed consensus (BFT consensus for short)
\cite{LSP82,pbft}, for both the permissioned setting
\cite{Gramoli17,Sousa18,hotstuff,libraSMR} and the public permissioned one
\cite{algorand,cosmos,tendermint,thunderella}. This paper focuses on the
alternatives in the permissioned and public permissioned setting. Thanks to the power of
consensus, this blockchains are able to guarantee so-called absolute finality, which
ensures that no value inserted into the ledger can ever be revoked. It has been shown that
these blockchains, even though rarely accompanied by formal distributed systems
specifications, can be built over standard distributed computing abstractions such as
state machine replication \cite{lamport78} (SMR for short) or repeated consensus
\cite{DDFPT08} (RC for short). Examples of such systems are Hyperledger \cite{Sousa18},
which builds a distributed ledger over SMR (in Hyperledger there is no notion of blocks), and Tendermint \cite{tendermint,cosmos}, Tenderbake \cite{tenderbake} and Algorand \cite{algorand}, which
build a blockchain over RC.

In case of a distributed ledger built on top of an SMR service (DL-SMR for short),
transactions are interpreted as commands that are submitted directly by clients to the SMR
service, while in case of a blockchain built on top of an RC service (BC-RC for short),
transactions are collected by replicas into blocks, which are the input values that
replicas propose to the RC service. Although the two approaches may seem equivalent at
first sight since both produce a sequence of commands/transactions, they differ however in
their ability to prevent \emph{censorship}. In DL-SMR the SMR service decides on the order
in which all the commands submitted by clients are eventually served. On the other hand,
In BC-RC the RC service makes a series of \textit{independent} decisions on each new block
that is committed to the chain. The result is that in BC-RC a transaction submitted by a
client could in principle be censored if replicas skip to include it into the blocks that
they input to the underlying RC service.\footnote{As described in \cite{Vit15}, censorship
  has an impact on financial mechanisms such as \emph{contracts for difference}, on
  \emph{auditable computation}, and even on committee-based BFT consensus protocols, where
  censorship may enable a collusion of validators to prevent other validators from joining
  the consensus pool.} Thus, client requests are never guaranteed to be committed to the
sequence, and we say that RC lacks  user fairness. Informally,  user fairness guarantees that a transaction requested by a
correct client will eventually be committed to the blockchain, unless the transaction
becomes invalid by the insertion of other transactions into the blockchain while the
requested transaction is still pending (see \S\ref{sec:dlsmr} for a formal definition). As
customary \cite{CKPS01,redbelly17,TendermintCorrectness,tenderbake} validity is defined
externally by a deterministic, application-dependent predicate.

The lack of  user fairness above is pervasive in the setting of blockchains, since
many of them are built on top of an RC service (\eg\
\cite{dissecting,tendermint,cosmos,libraSMR,hotstuff,algorand,thunderella} \AG{I included
  HotStuff, Libra and Thunderella here, shall we include others?}) and are therefore
potentially vulnerable to censorship. As a consequence, these blockchains can only ensure
user fairness under the assumption that all the proposers are correct. % asses on the
%importance of this problem, consider that

Some of these blockchains \cite{dissecting,tendermint,libraSMR,hotstuff} improve the
fairness guarantee by rotating the proposer, but this does not achieve  user
fairness, because there is no mechanism that ensures that an infinite number of correct proposers will be able to get their proposal agreed. %, only an asymptotic property that guarantees that the probability of a
%correct transaction being censored is inversely proportional to the frequence at which
%proposers are rotated.

\AD{Just to give an example, hunderds of applications and services are build on top of
  Cosmos (based on Tendermint)
  \url{https://cosmos.network/ecosystem/apps/?ref=cosmonautsworld} and Algorand
  \url{https://www.algorand.com/ecosystem/use-cases})}

\AG{Libra pipelines certified children in a tree, and decdides monotonically on the tree,
  which is not strictly an RC apporach because there's no new independent consensus
  instance. Other approaches like Streamlet are not RC-based either. Explore whether our
  solution works on Libra and Streamlet in future work? Can these be described by the
  framework in \cite{RZS10}?}

\AG{Add something on the formalisation being a contribution too, before moving to
  Problem.}

\AG{Something on the reduction and the optimised transformation.}

\paragraph{Problem} We tackle the problem of providing user fairness to
blockchains built on top of RC in a seamless way. By seamless we mean that the property
must be provided no matter the underlying consensus protocol implementing RC. Let us
stress that, since blockchains are running systems today, techniques to implement them as
replicated objects should be as lightweight and generic as possible to avoid to rewrite
from scratch current running BFT consensus implementations.

\paragraph{Contributions}

\AG{This needs rewritting, but is more mechanical than the paragraphs above, which should
  convey the scope and our solution.}

\AG{Case against BC-RC? Some reviewer identified that as a strenght. See how to do this
  diplomaticly.}

We systematise previous work by providing, for the first time, \textit{complete and
  rigorous} specifications (properties and algorithms) of both DL-SMR and BC-RC
constructions. We then prove that BC-RC does not solve DL-SMR in general, since the former
construction lacks the user fairness property of the latter. To cover this gap, we
introduce an abstraction for \emph{fair repeated consensus} (FRC for short) that refines
the RC abstraction, and we show that a correct implementation of FRC can be used to
implement a blockchain built on top of it (BC-FRC for short) with user fairness. We then
propose a reduction from FRC to RC, by aggregating \textit{sub-proposals} that a replica
hears of into a single proposal before inputting it to RC. We show that aggregating
sub-proposals from $f+1$ or more contributors provides a form of \emph{internal validity}
(as opposed to the external validity of
\cite{CKPS01,redbelly17,TendermintCorrectness,tenderbake}) that guarantees that the
decided block contains at least one transaction that was proposed by a correct replica,
which is an important ingredient to establishing that BC-FRC upholds user fairness.

Although the techniques used in this reduction are not entirely novel, our solution is, to the best
of our knowledge, the first one agnostic to the BFT consensus protocol in the underlying
implementation of RC, and thus truly generic.

%We further illustrate the application of our
%approach by applying the optimised transformation to two popular blockchains, Cosmos and
%Algorand, thus adding support for state machine replication to their consensus protocols
%without increasing their communication complexity.

\AG{The neccessary conditions let you check whether an implementation meets SMR. Cosmos,
  which is not equal to Tendermint (discuss) does not meet our nececssary conditions. Let
  be potilitical with Tendermint people, they're friends...}

\paragraph{Roadmap} In \S\ref{sec:preliminaries} we state the system model and
introduce some background on blockchains. In \S\ref{sec:constructions} we rigorously
specify the DL-SMR and BC-RC constructions, and we analyse the relation between them and
show that DL-SMR is strictly stronger that BC-RC. In order to close this gap, in
\S\ref{sec:achieving-fairness} we introduce the FRC abstraction and teh BC-FRC construction. We provide a reduction from FRC to RC and show that DL-SMR and BC-FRC
are equivalent.
%that adds a communication round, and we optimise this solution by introducing a
%communication-constant, generic transforamtion that takes any RC-BC based on BFT consensus
%with voting phases (virtually all the ``classical'' BFT consensus) and delivers a
%BC-FRC. The optimised solution uses pipelining to build the next certified block along the
%phases of each instance of consensus. In \S\ref{sec:examples} we instantiate the optimised
%solution to the two popular Tendermint and Algorand blockchains. 
In \S\ref{sec:related-work} we discuss previous definitions of
fairness and existing techniques to achieve user fairness. In \S\ref{sec:conclusion} by provide some concluding remarks.

\section{Preliminary definitions}\label{sec:preliminaries}

\subsection{System model}\label{sec:system-model}
We consider a message-passing distributed system composed of a potentially infinite set of
processes $\Pi$. A process can have any of the roles of \emph{client} or \emph{replica},
\ie, the set $\Pi$ can be partitioned into two (not necessarily disjoint) sets of clients
and replicas respectively, and we assume that the cardinality of the latter set of
replicas is always finite. The rest of our assumptions are conventional in the Byzantine
model. We assume that the network is \emph{partially synchronous}, \ie, the system
operates asynchronously until some unknown Global Stabilisation Time (GST), after which,
the system respects known and finite time bounds for computation and
communication. %\AD{i'm wondering if we need to specify the reliable broadcast here or
% not}
Processes are equipped with a cryptographic primitive to sign all messages they send and
we assume that signatures are not forgeable. In addition, we assume that channels are
reliable and thus processes communicate with each other through authenticated perfect
links. Processes may experience Byzantine failures, i.e., they may behave arbitrarily by
omitting to send and receive messages, and by altering the content of messages. Such
processes are said to be \emph{faulty}. Processes that are not faulty are said to be
\emph{correct}. We let $n$ and $f$ be respectively the cardinality of the set of replicas
and of the set of faulty replicas. We assume that $f$ is strictly smaller than $n/3$, \ie,
less than a third of the replicas are faulty.

\subsection{Blockchain data structure}\label{sec:blockchain}

A blockchain (BC for short) is a distributed ledger with append-only and non-repudiation
properties that consists of a replicated data structure composed of blocks. Each block
contains a finite set of transactions and is placed into a sequence called \emph{chain}.

Clients may interact with the BC by issuing requests to
\begin{enumerate*}[label=(\roman*)]
\item read the content of the blockchain, or
\item insert a new transaction $\tx$ into the blockchain.
\end{enumerate*}
We focus on the latter requests, which are the ones that impact how blocks are created and
inserted into the chain. %
% Since the requests to read have a trivial synchronisation complexity, we omit the reads
% and focus on the requests to insert a new transaction (transaction request).
A client issues a \emph{transaction request} for transaction $\tx$ by invoking the
$\request(\tx)$ primitive.

Without loss of generality, we abstract away the notion of transaction, which may be a set
of financial transactions or a set of operation calls on generic data. To this aim, we
consider an application-defined validity condition on the transactions contained in the
blocks that ensures that the every transaction is consistent with the semantics of the
application implemented on top of the blockchain (\eg, no double spending in
cryptocurrency blockchains) and that must be verified before a block is decided and
inserted into the chain.

We consider \emph{permissioned} blockchains, \ie, blockchains where the set of replicas
who are in charge of maintaining the data structure is known \emph{a priori},\footnote{This
  assumption can be lifted by considering \emph{committee-based} blockchains that allow
  for public, open access, where the set of participants is dynamically determined by a
  selection function that maps the current chain to a \emph{committee} of replicas that
  are in charge of producing the next block \cite{cosmos, tenderbake}.} which ensures
absolute finality (i.e., once a block is appended to the chain, it can never be
revoked). A necessary condition to obtain absolute finality is that each new block
appended to the chain is agreed upon with a BFT consensus protocol \cite{Bt-adt}. (Both
BC-RC and DL-SMR use BFT consensus and fall under this category.)

Correct clients are assumed to collectively issue an infinite number of transaction
requests, but no client (correct or faulty) can issue an infinite number of requests in a
finite period of time. Clients may issue transaction requests with a throughput that is
higher than the one at which transactions can be allocated in the blockchain. We make no
assumption on the amount of memory that each replica maintains, but we assume blocks to
have bounded size.

We write $\seq{b}$ to denote a chain that consists of an infinite sequence of blocks
$[\seqAt{b}{0},\seqAt{b}{1},\seqAt{b}{2},\ldots]$, where $\seqAt{b}{i}$ denotes the $i$-th
block in the chain and $\seqUpTo{b}{i}$ denotes the prefix of the chain up to the $i$-th
block. We sometimes decorate this notation with a replica subindex $r$ to denote the chain
(respectively the $i$-th block and the prefix of the chain) stored locally at replica $r$,
as in $\seqPr{b}{r}$, $\seqPrAt{b}{r}{i}$ and $\seqPrUpTo{b}{r}{i}$. %
% For the sake of simplicity we consider only the relevant set of information necessary
% for our presentation, i.e., each block contains a pointer to the previous block in the
% chain and a finite sequence of transactions. %
% % In real implementations blocks contain additional information that the pointer to the
% % previous block and a sequence of transactions, omitted in the text, because not
% % necessary for our model.}}
We write $\langle{\uparrow}b,[\tx_0,\tx_1,\ldots, tx_n]\rangle$ to denote a block that
points to block $b$ and contains the sequence of transactions $\tx_0$ to $\tx_n$. %
% \AD{what do we mean here with address? the pointer is not related to an address, we
% should say: the genesis block is known a priori by all replicas.}
We write $\seqUpTo{b}{i}\append\txs$ and
$\seqUpTo{b}{i}\concat[\langle{\uparrow}\seqAt{b}{i},\txs\rangle]$ as synonyms, which
denote the chain obtained by concatenating to the prefix $\seqUpTo{b}{i}$ a block that
contains the pointer ${\uparrow}\seqAt{b}{i}$ to the last block in the prefix and the
sequence of transactions $\txs$. %
% \AG{Perhaps use a notion of \emph{well-formedness} for chains?}

% \AD{Comment that the blocks contain more stuf in real implementations, for either the
% application or the consensus algorithm (nonce, \ldots). Comment on external validity
% also checking that information.}

The initial block $\seqAt{b}{0}=\langle\bot,[]\rangle$, called the \emph{genesis block},
is a special block that is known a priory by each correct replica and which contains an
undefined pointer and an empty sequence of transactions.

Throughout the paper, we emphasise the data structure maintained by the replicas, which
consists of the chains stored locally at each replica. We let the current chain $\seq{b}$
be the longest common prefix of the chains that are stored at correct replicas, and assume
that correct clients have at their disposal a deterministic primitive $read()$ that
retrieves the current chain, whose details are out of the scope of this paper.

We let each transaction to be uniquely identified, \eg, by letting it be timestamped and
signed by the client that issued it. We say that a transaction $\tx$ issued by a client is
\emph{finalised} iff a block whose sequence of transactions contains $\tx$ is inserted
into the current chain. Otherwise we say that transaction $\tx$ is \emph{pending}. %
% \footnote{Permissioned blockchains exhibit the property that once a block is produced,
% it will never be revoked, and therefore for our case in point all the transactions
% inserted in the chain are finalised.} \AG{Comment on the assumption on permissioned
% blockchains somewhere, perhaps in system model.} \AG{Assume that transactions may be
% labelled with the client that issued them and with a timestamp, as for each transaction
% to be uniquely identified regardless of its content.}
%
%Each transaction is signed by the client that issued it. %

% We say that a transaction $\tx$ is \emph{finalised} when the blockchain has produced a
% block $\seqAt{b}{i}$ such that $\tx$ is in the sequence of transactions contained in
% $\seqAt{b}{i}$, and the block $\seqAt{b}{i}$ will never be revoked. Permissioned
% blockchains exhibit the property that once a block is produced, it can never be revoked,
% and therefore for our case in point all the transactions inserted in the chain are
% finalised. \AG{Comment on the assumption on permissioned blockchains somewhere, perhaps
% in system model.}

\subsection{Reducibility among distributed problems}\label{sec:reducibility}
% We consider \emph{reducibility} among distributed problems as in
% \cite{CT1996,Milosevic11}, which we define formally as follows. \AG{Revise/update
%   citations and check definition of reducibility.} Let $A$ and $B$ be the specifications
% of two different distributed problems (\ie, two different distributed abstractions). We
% say that $A$ \emph{reduces to} $B$ iff the existence of a solution of $B$ entails the
% existence of a solution of $A$. (In practice, a reduction from $A$ to $B$ is witnessed by
% a correct implementation of $A$ that uses any correct implementation of $B$ as a
% blackbox.) %
We consider \emph{reducibility} among distributed problems as in Def.\ref{def:reduces-to} below. %
% in the spirit of \cite{CT1996,Milosevic11}. %
Consider two different distributed problems $A$ and $B$. Informally, $A$ can be reduced to $B$ if the existence of a solution of $B$ entails the existence of a solution of $A$. %
% More formally, we say that $A$ \emph{reduces to} $B$ iff there exists an algorithm that transforms the output of each
% process $p$ in any correct implementation of $B$ as to obtain a correct implementation of $A$. %
\begin{definition}\label{def:reduces-to}
Let $A$ and $B$ be two different distributed abstractions. $A$ \emph{reduces to} $B$ iff
there exists an algorithm that implements $A$ by making black-box use of $B$ and by using
asynchronous communication.
\end{definition}
The rationale behind this definition is that the algorithm that implements $A$ should not
rely on the partial synchronicity of the system, but only on the authenticated perfect
links provided by the process signatures and the asynchronous reliable channels in our
system model assumptions (see \S\ref{sec:system-model}). Since consensus is impossible in
a fully asynchronous setting due to the FLP impossibility \cite{flp}, this rationale
ensures that the algorithm implementing $A$ necessarily uses the black-box consensus
provided by $B$, instead of implementing consensus on its own.

We say that $B$ \emph{is at least as strong as} $A$ (written $A\subseteq B$) iff $A$
reduces to $B$ (\ie, existence of a correct implementation of $B$ guarantees the existence
of a correct implementation of $A$). We say $B$ is \emph{strictly stronger than} $A$
(written $A\subset B$) iff $A$ reduces to $B$ and $B$ is not reducible to $A$. We say $A$
and $B$ are \emph{equivalent} (written $A\equiv B$) iff $A$ reduces to $B$ and $B$ reduces
to $A$.

% \AD{The point is the following: so far, there are many blockchain definitions, mainly on
% the data structure itself, without looking at client fairness.In the practice, DL has
% been implemented either over RC or SMR providing different specifications. The gap
% between the two is user fairness, provided by DL over SMR. However most of the current
% DL solutions are over RC. We show how to get the same specification of a DL over SMR
% starting from a DL over RC.}  \SB{The main problem here is that we do not have a
% specification for the blockchain itself. So it is not clear what we are looking for. At
% the end, what we want is something as we specified for the SMR enriched with the
% validity property. which in this case, it more DL over SMR.}

%%% Local Variables:
%%% fill-column: 90
%%% require-final-newline: t
%%% mode-require-final-newline: t
%%% next-line-add-newlines: nil
%%% show-trailing-whitespaces: t
%%% indent-tabs-mode: nil
%%% ispell-dictionary: "british"
%%% mode: latex
%%% TeX-PDF-mode: t
%%% TeX-master: "main"
%%% End:

\section{Distributed ledger constructions}\label{sec:constructions}

In this section we provide rigorous specifications of the DL-SMR and BC-RC constructions.
We propose a general framework for building RC-BCs that use RC as a building block.

\subsection{Distributed ledger over SMR (DL-SMR)}
\label{sec:dlsmr}

We consider a variant of the state machine replication (SMR) abstraction in
\cite{synchHotstuffs,ANRX2020}. SMR specifies a strongly consistent replicated service
given by a deterministic state machine with a set of commands that can be issued by
clients. A set of replicas commits these commands into a linearisable log and produces a
consistent result of applying them, which is akin to the execution of the service by a
single correct replica that applies the commands in the order in which they occur in the
log.\footnote{\label{fn:linearisability}The definition of SMR in
  \cite{synchHotstuffs,ANRX2020} assumes that the operations can be linearised
  \cite{HW1990}, which imposes a one-to-one correspondence between the operations in a
  history and those in its linearisation. We reuse this intuition but note that
  linearisability cannot be lifted to a Byzantine setting just as readily
  \cite{MRL98,LisRod05}.}

Due to its nature, SMR can be directly used as a distributed ledger construction, to which
we refer as DL-SMR. \AG{Stress that the $\submit()$ primitive of SMR is a black box, whose
  details are out of the sope of the paper.} This can be achieved by viewing transactions
as commands and logs as states, and by considering that a transaction is finalised when it
is committed (\ie, appended) to the log.

We use the notation for sequences that we introduced in \S\ref{sec:blockchain} and write
$\seq{\ell}$ for the log, $\seqAt{\ell}{k}$ for the transaction at position $k$ of the
log, and $\seqUpTo{\ell}{k}$ for the prefix of the log up to position $k$. (Mind that a
position $k$ of the log refers to an individual transaction, as opposed to a block in a
blockchain, which may contain several transactions.) We may decorate the log $\ell$ with a
replica subindex $r$ to denote the log stored locally at replica $r$, as in
$\seqPr{\ell}{r}$, $\seqPrAt{\ell}{r}{k}$ and $\seqPrUpTo{\ell}{r}{k}$.

Clients issue transactions through the primitive $\submit(\tx)$, and a correct replica may
eventually commit some transaction $\tx$ for each position of the log given that $\tx$ is
valid to the replica. %
% \ie, it does not conflict with any of the transactions previously
% occurring in the log stored locally at the replica. %
Following the conventions on distributed ledgers
\cite{CKPS01,redbelly17,TendermintCorrectness,tenderbake}, we define validity of a
transaction $\tx$ by employing an application-defined predicate $\validTxs()$ that takes
the log resulting by appending $\tx$ to the current log and returns true iff the
transactions in the log are consistent with the semantics of the application implemented
on top of the distributed ledger. We say that a transaction $\tx$ is \emph{valid to
  replica $r$ at position $k$} iff $\validTxs(\seqPrUpTo{\ell}{r}{k-1}\concat[\tx])$
holds. We may omit the $k$ and the $r$ and write ``valid to replica $r$'' or just ``valid'' when they are clear from the context.

After a correct replica $r$ commits transaction $\tx$, the replica applies the transaction
by appending $\tx$ to the replica's current log $\seqPrUpTo{\ell}{r}{k}$. We omit any
notification to the client since we put the emphasis on the data structure (the log)
maintained by each replica.

% and it may as well notify the client that issued the
% transaction of the result state produced. We consider that the local state is the local
% current log $\seqPrUpTo{\ell}{r}{k}$ itself, and that applying transaction $\tx$ to this
% state amounts to appending the transaction to the log. We omit client notifications since
% we put the emphasis on the data structure (the log) maintained by each replica.

We require that DL-SMR meets the properties below:
\begin{description}
\item[{\rm (\emph{Safety})}] Two correct replicas do not commit different transactions at
  the same log position.
  % \item (\emph{User fairness}) If a correct client issues a transaction $\tx$ that is
  %   valid
  %   to some correct replica~$r$ at the moment when $\tx$ is issued, then at some
  %   posterior moment either $\tx$ becomes invalid to $r$, or otherwise $\tx$ is
  %   finalised.
  % \item (\emph{User fairness}) Let a correct client issue a transaction $\tx$ that is
  %   valid to some correct replica~$r$ at the moment when $\tx$ is issued. If $\tx$ never
  %   becomes invalid to $r$ at some posterior moment, then $\tx$ is eventually finalised.
\item[{\rm (\emph{User fairness})}] Let a correct client issue a transaction $\tx$ that is
  valid to some correct replica~$r$ at the moment when $\tx$ is issued. Either $\tx$
  becomes invalid to $r$ at some posterior moment, or otherwise $\tx$ is eventually
  finalised.
\item[{\rm (\emph{Log validity})}] If $\seqPrUpTo{\ell}{r}{k}$ is the current log of a
  correct replica~$r$, then $\validTxs(\seqPrUpTo{\ell}{r}{k})$ holds.
\item[{\rm (\emph{Log finality})}] A correct replica commits at most one transaction for
  each position in the log.
\end{description}

\AG{Stress the generality of our user fairness.}

\emph{Safety} above corresponds to the property with the same name in
\cite{synchHotstuffs,ANRX2020}. Our \emph{User fairness} above characterises when a
transaction issued by a correct client will be finalised, which differs from \emph{User
  fairness} in \cite{ace2019} which states that the quality of the chain is such that at
least one half of the transactions that it contains has been issued by correct
clients.\footnote{In \cite{ace2019} they assume the fully asynchronous model which we do
  not, and thus our specification and theirs are incomparable in strength.} When assuming
that every client is correct and that every transaction is valid to every replica---\ie,
$\validTxs$ is the constant predicate that always returns true---then our \emph{User
  fairness} coincides with \emph{Liveness} in \cite{synchHotstuffs,ANRX2020}.\footnote{We
  observe however that although \emph{Liveness} in \cite{synchHotstuffs} states that
  ``every request is evenutally committed by honest replicas'', Thm.4 of that work only
  proves liveness for honest replicas, stated as ``all honest replicas keep committing
  requests''.} %
% Since applying a transaction to a local state is a deterministic operation (\ie, the
% local state obtained by applying a given transaction to a given local state is always
% unique) a replica is allowed to commit any number of times times for the same log
% position given that the committed transaction is always the same, which is ensured by
% \emph{Safety}. We thus refrain ourselves from requiring an explicit \emph{Integrity}
% property on the lines of \cite{ace2019} that ensures that replicas commit at most once
% for each log position, and we implicitly assume that immediate finality holds.
% \AG{Do we prefer to add Integrity (\ie, Finality) explicitly?}
The last two properties of DL-SMR above correspond respectively to the \emph{Validity} and
\emph{Integrity} properties of \cite{ace2019}. %
% We have dubbed them \emph{Log validity} and \emph{Log finality} respectively to
% distinguish them from the analogous properties in RC above.

% % TODO: See how to thread this with the internal/external validity problem
% Since we put the emphasis on the data structure stored at the replicas (the log), the
% state machine underlying our DL-SMR has a trivial transition function that accepts every
% transaction and appends it to the state, which is the log. %
% % Therefore, the implicit notion of validity induced by this transition function is the
% % trivial one that considers that every transaction is always valid.
% However, in the Byzantine model the external validity condition is mandatory
% \cite{CKPS01}. Indeed, it is usually impossible to rule out that correct replicas commit
% spurious transactions created out of thin air by the sole effect of Byzantine
% replicas. %
% % \ST{Clarify this idea.}  However, in the Byzantine model it is usually impossible to
% % rule out that correct replicas commit spurious transactions created out of thin air by
% % the sole effect of Byzantine replicas, and thus the external validity condition is
% % mandatory in this model \cite{CKPS01}. %\ST{Clarify this idea.}

For simplicity, when considered in the context of blockchain implementations, we let the
logs in DL-SMR correspond to chains by letting each transaction in the log to be contained
in a different block. Formally, a log $\seqUpTo{\ell}{k}$ of length $k$ corresponds to a
chain $\seqUpTo{b}{k+1}$ of length $k+1$ where $\seqAt{b}{0}$ is the genesis block which
does not contain any transactions and which is known to every correct replica, and each
block $\seqAt{b}{i}$ with $i>0$ points to $\seqAt{b}{i-1}$ and contains the single
transaction $\seqAt{\ell}{i-1}$, which is the transaction at position $i-1$ in the log. To
wit, the log $[\tx_0,\tx_1]$ corresponds to the chain
$[\langle\bot,[]\rangle]\append[\tx_0]\append[\tx_1]$.

% Typically, blockchains aim at implementing the SMR abstraction, and many authors have
% claimed that some RC-based blockchains as the ones described above solve SMR
% \cite{synchHotstuffs,tendermint,dbft}. However, in many of these works the notion of SMR
% that they consider is somehow under-specified, and the claims contained in them are not
% definite.

Next, we discuss the alternative RC-BC construction, which is customary in the blockchain
setting. We start by presenting its RC building block.

\subsection{Repeated consensus (RC)}
\label{sec:rc}

We consider a variant of the repeated consensus (RC) abstraction in \cite{DDFPT08}. RC
consists of an infinite sequence of instances of BFT consensus which altogether decide an
infinite sequence of values. Without loss of generality, we will consider blocks as values
and let the sequence of decided values coincide with the chain. At each consensus instance
$i$, each correct replica inputs a block $b$ through the primitive $\inputBlock(i,b)$
(note that the block $b$ may be different from the inputs of other replicas) and
eventually outputs some decided block $b'$ which is appended to the local chain stored by
the replica. We let each replica sign the block that the replica inputs at each consensus instance, such that every participant could check the provenance of each decided block (the replica who proposed the block).

As before, we capture the validity of a decided block $b$ by employing an
application-defined predicate $\validBlocks()$ that takes the whole chain ending with $b$
and returns true iff the chain is well-formed and its transactions are consistent with the
semantics of the application implemented on top of the blockchain. We say that a block $b$
is \emph{valid to replica $r$ at consensus instance $i$} iff
$\validBlocks(\seqPrUpTo{b}{r}{i-1}\concat[b])$ holds. We may omit the $i$ and the $r$ and
write ``valid to replica $r$'' or just ``valid'' when they are clear from the context.

We require that RC meets the properties below:
\begin{description}
\item[{\rm (\emph{Agreement})}] If $\seqPr{b}{r}$ and $\seqPr{b}{s}$ are respectively the
  current outputs of two correct replicas~$r$ and~$s$, then either $\seqPr{b}{r}$ is a
  prefix of $\seqPr{b}{s}$, or $\seqPr{b}{s}$ is a prefix of $\seqPr{b}{r}$.
\item[{\rm (\emph{Termination})}] Every correct replica has an infinite output.
\item[{\rm (\emph{Chain validity})}] If $\seqPrUpTo{b}{r}{i}$ is the current output of a
  correct replica~$r$, then $\seqPrUpTo{b}{r}{i}$ is signed by some replica $r'$ (not necessarily different from $r$) and
  $\validBlocks(\seqPrUpTo{b}{r}{i})$ holds 
\item[{\rm (\emph{Chain finality})}] A correct replica outputs at most one block for each
  consensus instance.
\end{description}
\noindent The first two properties of RC adapt the properties with the same name in
\cite{DDFPT08}. \emph{Chain validity} states that each new block is well-formed and
upholds the semantics of the application implemented on top of the blockchain, and
\emph{Chain finality} collects the absolute finality assumption for permissioned
blockchains.

% It is easy to see that each instance of consensus in RC shares the interface and the
% properties of \emph{Agreement}, \emph{Termination} and \emph{Integrity} of Byzantine
% consensus in the textbook \cite[pp.244-246]{SDbook2011}, but it shares neither the
% \emph{Weak validity} nor the \emph{Strong validity} properties there. \ST{Comment on
% external validity being used to verify proposed values by receivers, which is also
% possible thanks to the authenticated links.}

% RC omits the clients and considers that, at each consensus instance, each correct replica
% inputs a block which may be different from the inputs of the other replicas. We are
% interested in blockchains in which replicas input blocks that contain the transactions
% issued by clients through the $\request(\tx)$ primitive. In order to formally characterise
% such blockchains, we introduce an abstraction that is parametric in a correct
% implementation of RC.

\subsection{Blockchain over RC (BC-RC)}
\label{sec:bcrc}

% \begin{figure}
%   \begin{minipage}[t]{.32\linewidth}
%     \begin{algorithm}[H]
%       \setcounter{AlgoLine}{0}
%       \small
%       \SubAlgo{\Procedure $\request(\tx)$\label{lin:bcrc-request}}
%       {
%         $\seqUpTo{b}{i}\assign{}$
%         $\readChain()$\label{lin:bcrc-retrieve-current}\;
%         \If{$\validBlocks(\seqUpTo{b}{i}\append[\tx])$}
%         {send $\langle\REQ,\tx\rangle$ to every replica\;\label{lin:bcrc-send}}
%       }
%       \caption{\\[2.7pt]BC-RC Client $c$}
%       \label{alg:rc-based-blockchain-client}
%     \end{algorithm}
%   \end{minipage}
%   \hfill
%   \begin{minipage}[t]{.67\linewidth}
%     \begin{algorithm}[H]
%       \setcounter{AlgoLine}{0}
%       \small
%       $\pool\assign []$\label{lin:bcrc-init}\tcp*{Initialisation}
%       \SubAlgo{\Upon received $\langle\REQ,\tx\rangle$ from a client
%         \label{lin:bcrc-receive-transaction}}
%       {
%         $\pool\assign\pool\concat[\tx]$\label{lin:bcrc-enlarge-pool}\;
%       }
%       \SubAlgo{\Upon $\RC.\outputBlock(i)$
%         \label{lin:bcrc-new-instance}}
%       {
%         clear transactions in $\RC.\seqPrUpTo{b}{r}{i}$ from
%         $\pool$\label{lin:bcrc-clear}\;
%         $\txs\assign$ pick from $\pool$ up to block
%         size and valid to $\RC.\seqPrUpTo{b}{r}{i}$\label{lin:bcrc-select}\;
%         $\RC.\inputBlock(i+1,\langle{\uparrow}
%         (\RC.\seqPrAt{b}{r}{i}),\txs\rangle)$\label{lin:bcrc-input}\;
%       }
%       \caption{BC-RC Replica $r$}
%       \label{alg:rc-based-blockchain-replica}
%     \end{algorithm}
%   \end{minipage}
% \end{figure}

A \emph{blockchain over repeated consensus} (BC-RC) is a blockchain implemented on top of
a correct implementation of RC, for which we require the four properties of
\emph{Agreement}, \emph{Termination}, \emph{Chain validity} and \emph{Chain finality}, and
which additionally meets the properties below:
\begin{description}
\item (\emph{Valid request}) If a correct client sends a request to insert a transaction
  $\tx$ to some replica, then $\validBlocks(\seqUpTo{b}{i}\append[\tx])$ holds where
  $\seqUpTo{b}{i}$ is the current chain at the moment when the client sends the request.
\item (\emph{Valid input}) If a correct replica inputs a block $b$ for some consensus
  instance, then for each transaction $\tx$ contained in $b$ there exists some client from
  which the replica received a request to insert $\tx$.
\end{description}
The two addtional properties above characterise the interaction between BC-RC and the RC
building block, not the behaviour that is externally observable from the output of
BC-RC. In particular, \emph{Valid request} enforces that the requests received from a
correct client are valid with respect to the current chain at issuing time,\footnote{We
  assume a deterministic primitive procedure $\readChain()$ at client side that retrieves
  the current chain.}  and \emph{Valid input} enforces that correct replicas do not input
spurious transactions created out of thin air. \AG{Comment on the current chain and the
  possibility of retrieving an old one, which would be similar to linearising your
  operation in a previous moment on time, and would meet the proeprties.}

Alg.\ref{alg:rc-based-blockchain-client}--\ref{alg:rc-based-blockchain-replica}
respectively collect the code for the client and the replica of a generic BC-RC, which
takes a correct implementation of RC as parameter. At a correct client, the primitive
$\request(\tx)$ (Lin.\ref{lin:bcrc-request} of Alg.\ref{alg:rc-based-blockchain-client})
first retrieves the current chain $\seqUpTo{b}{i}$ (Lin.\ref{lin:bcrc-retrieve-current})
and then sends a request $\langle\REQ,\tx\rangle$ to insert transaction $\tx$ to all the
replicas, given that the transaction is valid with respect to the current chain
(Lin.\ref{lin:bcrc-send}), thus upholding \emph{Valid request}.

Each correct replica initialises a pool of pending transactions to the empty sequence
(Lin.\ref{lin:bcrc-init} of Alg.\ref{alg:rc-based-blockchain-replica}). Upon the reception
of a request $\langle\REQ,\tx\rangle$ from a client, the replica collects the transaction
$\tx$ into the pool
(Lin.\ref{lin:bcrc-receive-transaction}--\ref{lin:bcrc-enlarge-pool}). Once a new block is
output at consensus instance~$i$ (Lin.\ref{lin:bcrc-new-instance}) the replica first
clears the finalised transactions in the chain up to that block
(Lin.\ref{lin:bcrc-clear}), and assembles into a proposed block some of the valid
transactions in the pool (not necessarily all, since the size of blocks is
bounded)\footnote{We intentionally under-specify how the valid transactions are picked and
  the invalid ones are cleared from the pool, which may be application-defined.} and moves
to consensus instance $i+1$ by inputting the proposed block
(Lin.\ref{lin:bcrc-select}--\ref{lin:bcrc-input}). \emph{Valid input} holds since the
transactions have been picked from the pool and the pool collects transactions received
from clients. The implementation of RC provides the local chain $\RC.\seqPr{b}{r}$, and
the primitive procedure $\RC.\inputBlock(i,b)$ and notification $\RC.\outputBlock(i)$. We
assume that every replica starts its execution by notifying the output of the genesis
block at consensus instance $0$.

Alg.\ref{alg:rc-based-blockchain-client}--\ref{alg:rc-based-blockchain-replica} witness  
the reduction from BC-RC to RC.

\begin{algorithm}[t]
    \setcounter{AlgoLine}{0}
    \small
    \SubAlgo{\Procedure $\request(\tx)$\label{lin:bcrc-request}}
    {
      $\seqUpTo{b}{i}\assign{}$
      $\readChain()$\label{lin:bcrc-retrieve-current}\;
      \lIf{$\validBlocks(\seqUpTo{b}{i}\append[\tx])$}
      {send $\langle\REQ,\tx\rangle$ to every replica\label{lin:bcrc-send}}
    }
    \caption{BC-RC Client $c$}
    \label{alg:rc-based-blockchain-client}
  \end{algorithm}
  \begin{algorithm}[t]
    \setcounter{AlgoLine}{0}
    \small
    $\pool\assign []$\label{lin:bcrc-init}\tcp*{Initialisation}
    \SubAlgo{\Upon received $\langle\REQ,\tx\rangle$ from a client
      \label{lin:bcrc-receive-transaction}}
    {
      $\pool\assign\pool\concat[\tx]$\label{lin:bcrc-enlarge-pool}\;
    }
    \SubAlgo{\Upon $\RC.\outputBlock(i)$
      \label{lin:bcrc-new-instance}}
    {
      clear transactions in $\RC.\seqPrUpTo{b}{r}{i}$ from
      $\pool$\label{lin:bcrc-clear}\;
      $\txs\assign$ pick from $\pool$ up to block
      size and valid to $\RC.\seqPrUpTo{b}{r}{i}$\label{lin:bcrc-select}\;
      $\RC.\inputBlock(i+1,\langle{\uparrow}
      (\RC.\seqPrAt{b}{r}{i}),\txs\rangle)$\label{lin:bcrc-input}\;
    }
    \caption{BC-RC Replica $r$}
    \label{alg:rc-based-blockchain-replica}
\end{algorithm}

\begin{lemma}\label{lem:bcrc-reducible-rc}
  BC-RC is reducible to RC.
\end{lemma}
\begin{proof}
  We first notice that none of the steps of the generic BC-RC of
  Alg.\ref{alg:rc-based-blockchain-client}--\ref{alg:rc-based-blockchain-replica} rely on partially
  synchronous communication. It is easy to prove that this generic BC-RC meets \emph{Valid
  requests} by Lin.\ref{lin:bcrc-send} of Alg.\ref{alg:rc-based-blockchain-client}, and that it
  meets \emph{Valid input} since all the requests are signed by the clients and signatures are not
  forgeable. \emph{Agreement}, \emph{Termination}, \emph{Chain validity}, and \emph{Chain finality}
  hold respectively by the properties with the same name of the correct implementation of RC that
  is taken as parameter.
\end{proof}

% \begin{figure*}
%   \begin{center}
%     \begin{tikzpicture}{\small
%         \node[abstraction] at (0,-.6) {
%           \parbox{6.9cm}{
%             DL-SMR ($\textbf{in:}~\client.\submit(\tx) \mid \textbf{out:}~\seq{\ell}$)\\
%             {Safety, User fairness, Log validity, Log finality.}
%           }
%         } ;
%         \node at (8,-1.1) {{\large $\supset$}} ;
%         \node[abstraction] at (8.8,0) {
%           \parbox{6.9cm}{
%             BC-RC ($\textbf{in:}~\client.\request(\tx) \mid \textbf{out:}~\seq{b}$)\\
%             {Valid request, Valid input.}
%             \\ \\ \\[5pt]
%           }
%         } ;
%         \node[abstraction] at (8.9,-1.1) {
%           \parbox{6.7cm}{
%             RC ($\textbf{in:}~\replica.\inputBlock(b) \mid \textbf{out:}~\seq{b}$)\\
%             {Agreement, Termination, Chain validity, Chain finality.}
%           }
%         } ; }
%     \end{tikzpicture}
%   \end{center}
%   \caption{Relation among DL-SMR, BC-RC and RC (symbol $\supset$ indicates that BC-RC
%     reduces to DL-SMR but not the converse).}
%   \label{fig:relation-dlsmr-bcrc-rc}
% \end{figure*}

\begin{figure*}
  \begin{center}
    \begin{tikzpicture}{\small
        \node[abstraction] at (0,.8) {
          \parbox{3.5cm}{
            DL-SMR\\
            ($\textbf{in:}~\client.\submit(\tx) \mid \textbf{out:}~\seq{\ell}$)\\
            {Safety, User fairness,\\Log validity, Log finality.}
          }
        } ;
        \node at (4.6,0) {{\large $\supset$}} ;
        \node[abstraction] at (5.4,.95) {
          \parbox{3.5cm}{
            BC-RC\\
            ($\textbf{in:}~\client.\request(\tx) \mid \textbf{out:}~\seq{b}$)\\
            {Valid request, Valid input,\\Agreement, Termination,\\Chain validity, Chain finality.}
          }
        } ;
        \node at (10,0) {{\large $\subseteq$}} ;
        \node[abstraction] at (10.8,.8) {
          \parbox{3.5cm}{
            RC\\
            ($\textbf{in:}~\replica.\inputBlock(b) \mid \textbf{out:}~\seq{b}$)\\
            {Agreement, Termination,\\Chain validity, Chain finality.}
          }
        } ; }
    \end{tikzpicture}
  \end{center}
  \caption{Relation among DL-SMR, BC-RC and RC.}
  \label{fig:relation-dlsmr-bcrc-rc}
\end{figure*}
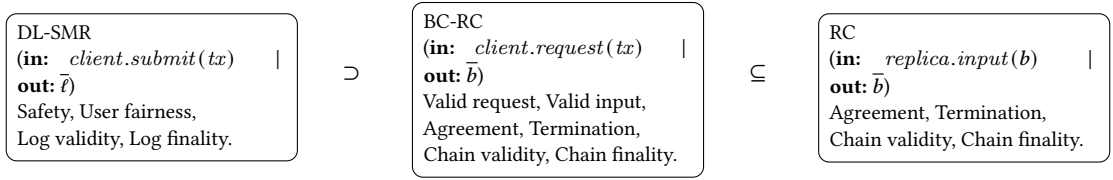

We next study in depth the relation between DL-SMR and BC-RC.

% In the next section we study in depth the relation between BC-RC and DL-SMR as defined
% here, and provide sufficient conditions for an RC-BC to implement DL-SMR.

% and clear invalid transactions

\subsection{The relation between DL-SMR and BC-RC}
\label{sec:relation-rc-smr}

Although both the BC-RC and the DL-SMR constructions implement a distributed ledger, the
two constructions differ in the key property of user fairness. %
% \AD{If we change the titles, then we should rephrase: Although both BC over RC and BC
% over SMR are two ways of providing blockchain solution, the two solutions differ in the
% key property that we discuss next}.
BC-RC focuses on the blocks being output at each consensus instance (the data structure),
and places no relation between each committed block and how the block was produced or by
whom (which clients issue which transactions, and how these are assembled into blocks by
replicas). By contrast, DL-SMR places such a relation and assumes that every transaction
issued by a correct client will either become invalid at some moment, or otherwise it will
be committed by all correct replicas. This is a stronger property cannot be achieved by
RC-BC alone.

% This stronger property warrants the reduction from
% BC-RC to DL-SMR, as stated by Lem.\ref{lem:bcrc-reducible-smr} below.

Lem.\ref{lem:bcrc-reducible-rc}--\ref{lem:smr-not-reducible-bcrc} below state that DL-SMR
is strictly stronger the RC-BC, since RC-BC reduces to DL-SMR but not the converse.

% \AG{Recall that the notion of reducibility requires a ``wait-free implmenetation of $X$
% in terms of $Y$'' \cite{CT1996,Her1991}.}

\begin{lemma}
  \label{lem:bcrc-reducible-smr}
  BC-RC is reducible to DL-SMR.
\end{lemma}
\begin{proof}
 Assume a correct implementation of DL-SMR. We will show how to produce a correct implementation of BC-RC that uses the implementation of DL-SMR as a black box and that does not rely on partially synchronous communication. We fix a genesis block $\seqAt{b}{0}$ and let every correct replica output it. Upon the invocation of $\request(\tx)$ by a correct client, we let the client
  trigger the primitive $\submit(\tx)$ of DL-SMR. We let the predicate $\validBlocks$
  coincide with $\validTxs$ when the log produced by SMR is viewed as a chain as in the
  correspondence detailed in~\S\ref{sec:dlsmr}. It suffices to show that
  \emph{Agreement}, \emph{Termination}, \emph{Chain validity} and \emph{Chain finality}
  hold, since \emph{Valid request} and \emph{Valid input} are internal properties of
  BC-RC, not observable from the output of the abstraction. \emph{Agreement}, \emph{Chain
    finality} and \emph{Chain validity} hold respectively by \emph{Safety}, \emph{Log
    validity} and \emph{Log finality}. It is easy to see that \emph{Termination} follows
  by the relaible channels and by \emph{User fairness} since, collectively, the correct clients issue an infinite
  number of transactions, and a correct client issues $\request(\tx)$ only if the
  transaction $\tx$ is valid with respect to the current chain, which upholds the
  $\validBlocks$ predicate. It could happen that for each transaction issued by a correct
  client, there exists a Byzantine client with good bandwidth that issues another
  transaction that invalidates the one from the correct client, preventing the latter from
  ever being committed. Although this very improbable scenario is pathological,
  \emph{Termination} would still hold because the output of correct replicas will be
  infinite since we assume that correct clients collectively issue an infinite number of
  transactions.
\end{proof}

\begin{comment}

 \begin{proof}[Sketch of the proof]
  A correct implementation of RC-BC can be produced from a correct implementation of
  DL-SMR by letting every correct invocation of $\request(\tx)$ trigger the primitive
   $\submit(\tx)$ of DL-SMR, and by letting the predicate $\validBlocks$ coincide with
   $\validTxs$ when the log produced by DL-SMR is viewed as a chain as in the
   correspondence detailed in~\S\ref{sec:preliminaries}. \emph{Agreement}, \emph{Chain
     finality} and \emph{Chain validity} hold respectively by \emph{Safety}, \emph{Log
    validity} and \emph{Log finality} of DL-SMR, and \emph{Termination} follows by
  \emph{Fairness} of DL-SMR since, collectively, the correct clients issue an infinite
  number of transactions.
\end{proof}
\end{comment}
% The relation that DL-SMR places between the committed transactions and the clients that
% produce them goes beyond what BC-RC can ensure. Lem.\ref{lem:smr-not-reducible-bcrc} below
% tightens the result stated by Lem.\ref{lem:bcrc-reducible-smr} and states that DL-SMR is
% at least as hard as BC-RC: no reduction is possible from the former to the latter without
% adding any synchronisation complexity since a correct implementation of BC-RC could always
% neglect the requests of a particular correct client, such that the transactions issued by
% that client would never be finalised thus violating \emph{User fairness}.

\begin{lemma}
  \label{lem:smr-not-reducible-bcrc}
  DL-SMR is not reducible to BC-RC.
\end{lemma}
\begin{proof}
%  We show that there exists not an algorithm that transforms a correct implementation of DL-SMR
%  We show that producing an implementation of DL-SMR that upholds \emph{User fairness}
%  from a correct implementation of BC-RC is
%  not possible in general: the implementation of BC-RC could always neglect the requests
%  of a particular correct client, such that the transactions issued by that client would
%  never be finalised. %
%
%  We show that there exists not an algorithm that implements DL-SMR by using a correct implementation of BC-RC
%   as a black box and without relying on partially synchronous communication. Upholding \emph{User fairness}
%   would not be possible in general because the implementation of BC-RC could always neglect the requests
%   of a particular correct client, such that the transactions issued by that client would
%   never be finalised.

  We show that there exists not an algorithm that implements DL-SMR by using a correct
  implementation of BC-RC as a black box and without relying on partially synchronous
  communication. Such an algorithm could not rely on the consensus provided by the implementation
  of BC-RC because then \emph{User fairness} would not hold in general, since the implementation
  of BC-RC could always neglect the transactions issued by a particular correct client, which would
  never be finalised. Therefore, the algorithm would need to implement consensus on its own,
  but then it would need to rely on partially synchronous communication, since asynchronous consensus is impossible due to the FLP impossibility \cite{FLP85}.
\end{proof}

\begin{comment}
\begin{proof}[Sketch of the proof]
  Producing an implementation of DL-SMR that upholds \emph{User fairness} from a correct
  implementation of BC-RC is not possible in general: the implementation of BC-RC could
  censor the requests of some correct client, such that the transactions issued by that
  client would never be finalised.
\end{proof}
\end{comment}

Figure~\ref{fig:relation-dlsmr-bcrc-rc} summarises the content of
Lem.\ref{lem:bcrc-reducible-smr}--\ref{lem:smr-not-reducible-bcrc}.

\section{Achieving user fairness}
\label{sec:achieving-fairness}

Lem.\ref{lem:smr-not-reducible-bcrc} from \S\ref{sec:preliminaries} sustains that BC-RC
does not solve DL-SMR in general, since BC-RC does not uphold the \emph{User fairness}
property. To cover this gap, we introduce the \emph{Fair Repeated Consensus} abstraction
(FRC for short), which refines RC, and the construction for building a blockchain over FRC
(BC-FRC for short).

\subsection{Fair repeated consensus (FRC)}

In order to formalize \emph{fair repeated consensus} (FRC), we need to introduce the
following preliminary definitions.  We consider an external $\validBlocks$ predicate as in
RC of \S\ref{sec:constructions} and, additionally, we let the external
$\fusion(\seqUpTo{b}{i},(\seq{t_j})_{j\in I})$ be a deterministic, application-defined
function that takes the current chain $\seqUpTo{b}{i}$ and a collection of sequences of
transactions $(\seq{t_j})_{j\in I}$, and fuses the collection into a result
sequence~$\seq{t}$ that contains valid transactions.

The FRC abstraction that we introduce below refines the RC abstraction from
\S\ref{sec:rc}. FRC shares interface and properties with RC, and additionally meets the
following:
\begin{description}

% \item[{\rm (\emph{Fair fusion})}] Let $\seqUpTo{b}{i}$ be a chain, $(\seq{t_j})_{j\in I}$
%  a collection of sequences of transactions, and $\seq{t}$ be the result sequence returned
%  by $\fusion(\seqUpTo{b}{i},(\seq{t_j})_{j\in I})$. The $\fusion$ function is such that
%  \begin{enumerate}[label=(\roman*)]
%  \item for every transaction $\tx$ that
%    occurs in some sequence $\seq{t_j}$ in the collection, either $\tx$ occurs in the result
%    $\seq{t}$, or otherwise there exists a position $k$ of $\seq{t}$ such that
%      \begin{displaymath}
%        \validBlocks(\seqUpTo{b}{i}\append(\seqUpTo{t}{k}\concat[\tx]))
%      \end{displaymath}
%      does not hold, and
%  \item for every transaction $\tx$, if there exists a sequence $\seq{t_j}$ in the collection and a position $k_j$      of that seqeunce such that $\validBlocks(\seqUpTo{b}{i}\append(\seqUpTo{t_j}{k_j}\concat[\tx]))$ does not hold, then there exists a position $k$ of the result $\seq{t}$ such that
%  \begin{displaymath}
%    \validBlocks(\seqUpTo{b}{i}\append(\seqUpTo{t}{k}\concat[\tx]))
%  \end{displaymath}
%  does not hold.
%  \end{enumerate}

\item[{\rm (\emph{Fair fusion})}] Let $\seqUpTo{b}{i}$ be a chain, $(\seq{t_j})_{j\in I}$
   a collection of sequences of transactions, and $\seq{t}$ be the result sequence returned
   by $\fusion(\seqUpTo{b}{i},(\seq{t_j})_{j\in I})$. For every transaction $\tx$ that
   occurs in some sequence $\seq{t_j}$ in the collection, either $\tx$ occurs in the result
   $\seq{t}$, or otherwise there exists a position $k$ of $\seq{t}$ such that
   $\validBlocks(\seqUpTo{b}{i}\append(\seqUpTo{t}{k}\concat[\tx]))$ does not hold.
\item[{\rm (\emph{Fusion validity})}] If a correct replica~$r$ outputs a block at some
  consensus instance $i$, then the sequence of transactions contained in that block
  coincides with the fusion of the sequences in the input blocks of $f+1$ or more
  replicas.
\end{description}

% \AG{Comment on how to include the information needed to check \emph{Fusion validity}
% (\ie, the signatures of sub-proposals) within a block in real implementations of
% blockchains.}

\emph{Fair fusion} collects the intuition that the fusion function should include in its
result the transactions in the collection of sequences that uphold the external validity
predicate.

\emph{Fusion validity} adapts \emph{Vector validity} from \cite{DS97} to our setting and
relaxes the threshold of contributors $2f+1$ in \cite{DS97} to $f+1$ (\ie, each output
block must contain transactions coming from at least $f+1$ contributors). \AG{Polish and
  show that FRC refines RC and comment on the different validity properties, and connect
  to internal validity.} Notice that \emph{Fusion validity} abstracts from the way the
inputs from replicas are aggregated into a result by using the fusion function, while
\emph{Vector validity} imposes that the inputs of the contributors are atomic, and that
the result contains at most one value for each replica in the system. %
%\begin{color}{red}
%  It is easy to show that \emph{Fusion validity} entails the usual \emph{Weak validity}
%  property in the textbook \cite[p.245]{SDbook2011} when the content of a block is atomic
%  (akin to a sequence with a single transaction) and when the fusion function takes any
%  valid input from the collection. As it is the case for \emph{Vector validty} in
%  \cite{DS97}, a strengthened version of \emph{Fusion validity} that requires
%  contributions from $2f+1$ or more replicas would also entail the usual \emph{Strong
%  validity} property \cite[p.246]{SDbook2011} when the content of a block is atomic and
%  when the fusion function takes the valid input from the collection that appears $f+1$
%  or more times, if such an input exists, or otherwise takes an undefined sequence
%  $\bot$.
%\end{color}

The fusion function and the validity predicate of FRC could be instantiated to cover a
varied spectrum of order-fairness models. At one end of the spectrum we could place the rather
general model of Bitcoin \cite{Nakamoto2008}, which considers no transaction ordering
within a block (akin to a set of transactions where every transaction commutes with each
other) and where causality of a transaction is only with respect to the previous blocks
(\ie, the predicates $\forall\tx\in\txs.~\validBlocks(\seqUpTo{b}{i}\append[\tx])$ and
$\validBlocks(\seqUpTo{b}{i}\append\txs)$ are synonymous). At the other end of the
spectrum we could place the rather specific model of Aequitas \cite{KMZGSJA20}, which
considers a transaction ordering that contains no loops except in some fragments of the
sequence where every transaction is output ``before or at the same time as'' each other,
and which best approximates \emph{receive-order} while avoiding the Condorcet paradox.

  \begin{algorithm}[t]
    \setcounter{AlgoLine}{0}
    \small
    $\done\assign 0$\tcp*{Initialisation}
    \SubAlgo{\Procedure $\inputBlock(i,b)$\label{lin:frc-input}}
    {
      $\txs\assign$ transactions in $b$\;
      send $\langle\SUBPROP,r,i,\txs\rangle$ to every replica;
    }
    \SubAlgo{\Upon received $\langle\SUBPROP,s,i,\txs_s\rangle$ from each
      $s\in S$ with $|S|\geq f+1$ and $\done = i - 1$\
      \label{lin:frc-receive-sub-proposals}}
    {
      $\txs\assign\fusion(\RC.\seqPrUpTo{b}{r}{i-1},(\txs_s)_{s\in S})$\;
      \label{lin:frc-aggregate}
      $\done\assign i$\;\label{lin:frc-done}
      $\RC.\inputBlock(i,\langle{\uparrow}
      (\RC.\seqPrAt{b}{r}{i-1}),\txs\rangle)$\;\label{lin:frc-input-rc}
    }
    \SubAlgo{\Upon $\RC.\outputBlock(i)$\label{lin:frc-output-rc}}
    {
      trigger $\outputBlock(i)$\;\label{lin:frc-output}
    }
    \caption{FRC Replica $r$}
    \label{alg:frc-replica}
  \end{algorithm}

FRC can be achieved with a correct implementation of RC. One direct way to reduce FRC to RC is to
add an extra communication round that performs the aggregation of the sub-proposals that are input
at FRC into an aggregated proposal that is input at the RC component. Alg.\ref{alg:frc-replica}
above witnesses such a reduction by implementing FRC on top of a correct implementation of RC. When
a replica invokes the primitive $\inputBlock(i,b)$ at FRC (Lin.\ref{lin:frc-input} of
Alg.\ref{alg:frc-replica}), it sends a sub-proposal for block $b$ to every other replica. After
receiving $f+1$ or more sub-proposals at consensus instance $i$
(Lin.\ref{lin:frc-receive-sub-proposals}) a replica aggregates the received sub-proposals into a
single proposal (Lin.\ref{lin:frc-aggregate}), marks that the aggregation at consensus instance $i$
has been done (Lin.\ref{lin:frc-done}), and inputs a block with the aggregated proposal at the RC
component (Lin.\ref{lin:frc-input-rc}). We assume that $\fusion$ in line \ref{lin:frc-input-rc}) is
such that \emph{Fair fusion} holds. We require the $\validBlocks$ predicate of the RC component to
be extended with a check that the input block has been contributed by $f+1$ or more replicas. This
check can be achieved by collecting the signed sub-proposals that have been aggregated together
with the aggregated proposal in the input block.\footnote{Checking that the input block has been
contributed by enough replicas can be optimised by letting the replicas sign each individual
transaction within a sub-proposal, or by tayloring some multi-signature scheme.} Once the RC
component triggers the notification that the block at consensus instance $i$ has been output, the
replica forwards this notification at FRC (Lin.\ref{lin:frc-output-rc}--\ref{lin:frc-output}).

\begin{lemma}
  \label{lem:frc-reducible-rc}
  FRC is reducible to RC.
\end{lemma}
\begin{proof}
  We first notice that none of the steps of Alg.\ref{alg:frc-replica} rely on partially synchronous
  communication. We show that Alg.\ref{alg:frc-replica} implements FRC. \emph{Agreement},
  \emph{Termination}, \emph{Chain validity} and \emph{Chain finality} hold respectively by the
  properties with the same name of the correct implementation of RC taken as parameter. It is
  enough to show that \emph{Fair fusion} and \emph{Fusion validity} hold. At each consensus
  instance $i$, \emph{Chain validity} ensures that the decided block $\seqAt{b}{i}$ is signed by
  some correct replica $r$. By Lin.\ref{lin:frc-receive-sub-proposals} of
  Alg.\ref{alg:frc-replica}, replica $r$ aggregated contributions from $f+1$ or more
  sub-proposals, and \emph{Fusion validity} holds. By Lin.\ref{lin:frc-aggregate}, replica $r$
  used the $\fusion$ function for aggregating the sub-proposals, and \emph{Fair fusion} holds and
  we are done.
\end{proof}

\AG{Comment on the size of the resulting block.}

\subsection{Blockchain over FRC (BC-FRC)}\label{sec:BC-FRC}

 \begin{algorithm}[t]
    \setcounter{AlgoLine}{0}
    \small
    \SubAlgo{\Procedure $\request(\tx)$\label{lin:bcfrc-request}}
    {
      $\seqUpTo{b}{i}\assign{}$
      $\readChain()$\label{lin:bcfrc-retrieve-current}\;
      \lIf{$\validBlocks(\seqUpTo{b}{i}\append[\tx])$}
      {send $\langle\REQ,\tx\rangle$ to every replica\label{lin:bcfrc-send}}
    }
    \caption{BC-FRC Client $c$}
    \label{alg:frc-based-blockchain-client}
  \end{algorithm}

  \begin{algorithm}[t]
    \setcounter{AlgoLine}{0}
    \small
    $\pool\assign []$\label{lin:bcfrc-init}\tcp*{Initialisation}
    \SubAlgo{\Upon received $\langle\REQ,\tx\rangle$ from a client
      \label{lin:bcfrc-receive-transaction}}
    {
      $\pool\assign\pool\concat[\tx]$\label{lin:bcfrc-enlarge-pool}\;
    }
    \SubAlgo{\Upon $\FRC.\outputBlock(i)$
      \label{lin:frc1bc-new-instance}}
    {
      clear transactions in $\FRC.\seqPrUpTo{b}{r}{i}$ from
      $\pool$\label{lin:frcbc-clear}\;
      $\txs\assign$ pick in FIFO order from $\pool$ up to block size
      and valid to $\FRC.\seqPrUpTo{b}{r}{i}$\label{lin:frc1bc-pick-fifo}\;
      $\FRC.\inputBlock(i+1,\langle{\uparrow}
      (\FRC.\seqPrAt{b}{r}{i}),\txs\rangle)$\label{lin:frc1bc-input}\;
    }
    \caption{BC-FRC Replica $r$}
    \label{alg:frc-based-blockchain-replica}
  \end{algorithm}

A \emph{blockchain over fair repeated consensus} (BC-FRC for short) is a blockchain
implemented on top of a correct implementation of FRC, for which we require the five
properties of \emph{Agreement}, \emph{Termination}, \emph{Chain finality}, \emph{Fair
  fusion} and \emph{Fusion validity}. \AG{Refer to the reduction in previous subsection.}
We also require the \emph{Valid request} and \emph{Valid input} properties of BC-RC in
\S\ref{sec:bcrc}, together with the additionally properties below:
\begin{description}
\item[{\rm (\emph{Request agreement})}] If a correct client issues a transaction $\tx$ then
  every correct replica eventually receives the request to insert $\tx$.
\item[{\rm (\emph{Stubborn input})}] If a correct replica~$r$ receives a request to insert
  a pending transaction $\tx$ at consensus instance $i$ then there exists a consensus
  instance $i_0\geq i$ such that
  \begin{enumerate}[label=(\roman*)]
  \item $\tx$ is already finalised at consensus instance $i_0$, or
%  \item for each $i'\geq i_0$ either the replica $r$ inputs a block that contains $\tx$ at
%    consensus instance $i'$, or otherwise  there exists a position $k_{i'}$ such that
%    $\validBlocks(\seqPrUpTo{b}{r}{i'-1}\append(\seqUpTo{\txs}{k_{i'}}\concat [\tx]))$
%    does not hold, where $\txs$ is the sequence of transactions in the block that
%    replica $r$ inputs at consensus instance $i'$.
   \item there exists a position $k$ such that
     \begin{displaymath}
         \validBlocks(\seqPrUpTo{b}{r}{i_0-1}\append(\seqUpTo{\txs}{k}\concat [\tx]))
     \end{displaymath}
     does not hold, where $\txs$ is the sequence of transactions in the block that the
     replica $r$ inputs at consensus instance $i_0$, or
   \item for every $i'\geq i_0$ the replica $r$ inputs a block that contains $\tx$ at
     consensus instance $i'$.
  \end{enumerate}
\end{description}

% \begin{figure}
%   \begin{algorithm}[H]
%     \setcounter{AlgoLine}{0}
%     \small
%     \SubAlgo{\Procedure $\request(\tx)$\label{lin:bcrc-request}}
%     {
%       $\seqUpTo{b}{i}\assign{}$
%       $\readChain()$\label{lin:bcrc-retrieve-current}\;
%       \lIf{$\validBlocks(\seqUpTo{b}{i}\append[\tx])$}
%       {send $\langle\REQ,\tx\rangle$ to replicas\label{lin:bcrc-send}}
%     }
%     \caption{BC-FRC Client $c$}
%     \label{alg:frc-based-blockchain-client}
%   \end{algorithm}
%   \begin{algorithm}[H]
%     \setcounter{AlgoLine}{0}
%     \small
%     $\pool\assign []$\label{lin:bcrc-init}\tcp*{Initialisation}
%     \SubAlgo{\Upon received $\langle\REQ,\tx\rangle$ from a client
%       \label{lin:bcrc-receive-transaction}}
%     {
%       $\pool\assign\pool\concat[\tx]$\label{lin:bcrc-enlarge-pool}\;
%     }
%     \SubAlgo{\Upon $\FRC.\newConsensusInstance(i)$
%       \label{lin:frc1bc-new-instance}}
%     {
%       clear transactions in $\FRC.\seqPrUpTo{b}{r}{i-1}$ from
%       $\pool$\label{lin:frc1bc-clear}\;
%       $\txs\assign$ pick in FIFO order from $\pool$ up to block size
%       and valid to $\FRC.\seqPrUpTo{b}{r}{i-1}$
%       and clear invalid transactions\label{lin:frc1bc-pick-fifo}\;
%       $\FRC.\inputBlock(i,\langle{\uparrow}
%       (\FRC.\seqPrAt{b}{r}{i-1}),\txs\rangle)$\label{lin:frc1bc-input}\;
%     }
%     \caption{BC-FRC Replica $r$}
%     \label{alg:frc1-based-blockchain-replica}
%   \end{algorithm}
% \end{figure}

Alg.\ref{alg:frc-based-blockchain-client}--\ref{alg:frc-based-blockchain-replica} above
introduce the client and the replica of a generic BC-FRC, which takes a correct implementation of
FRC as parameter.

The FRC parameter ensures \emph{Fair fusion} and \emph{Fusion validity}. Besides requiring
an implementation of FRC, the generic BC-FRC only differs form
Alg.\ref{alg:rc-based-blockchain-client}--\ref{alg:rc-based-blockchain-replica} in
Lin.\ref{lin:frc1bc-pick-fifo} of Alg.\ref{alg:frc-based-blockchain-replica}, which
enforces that the replicas pick the transactions from the pool in FIFO order, thus
upholding \emph{Stubborn input}.\footnote{Differently from
  Alg.\ref{alg:rc-based-blockchain-replica} in \S\ref{sec:preliminaries},
  Alg.\ref{alg:frc-based-blockchain-replica} here specifies a straightforward mechanism
  for handling pending transactions: traverse the pool in FIFO order while picking valid ones. We
  are however agnostic to other more speculative mechanisms for handling pending transactions that
  the blockchain model may stipulate, which ought to preserve \emph{Stubborn Input}.} Notice that
  \emph{Request agreement} holds trivially by the assumption of a reliable network. 
% \AG{How can we ensure that Lin.\ref{lin:frc1bc-clear} of
% Alg.\ref{alg:frc1-based-blockchain-replica} runs effectively if the pool can be
% arbitrarily long? Can our solution work if replicas have bounded memory?}

Alg.\ref{alg:frc-based-blockchain-client}--\ref{alg:frc-based-blockchain-replica} witness the reduction from BC-FRC to FRC.
\begin{lemma}\label{lem:bcfrc-reducible-frc}
  BC-FRC is reducible to FRC.
\end{lemma}
\begin{proof}
  The proof goes along the same lines than the proof of Lem.\ref{lem:bcrc-reducible-rc} in
  \S\ref{sec:bcrc}, where the properties of \emph{Valid request}, \emph{Valid input}, \emph{Fair
  fusion},\emph{Fusion validity}, \emph{Agreement}, \emph{Termination}, \emph{Chain validity},
  \emph{Chain finality} are straightforward to prove. \emph{Request agreement} holds trivially by
  the assumption of a reliable network, and it suffices to show that
  Alg.\ref{alg:frc-based-blockchain-client}--\ref{alg:frc-based-blockchain-replica} enjoy
  \emph{Stubborn input}. Let $\tx$ be a transaction received by replica $r$ at some consensus
  instance $i$. If $\tx$ is finalised at some consensus instance $i_0\geq i$, then (i) of
  \emph{Stubborn input} holds. If $\tx$ becomes invalid with respect to the sequence $\txs$ (up to
  some position $k$) that replica $r$ inputs at some consensus instance $i_0\geq i$, then (ii)
  of \emph{Stubborn input} holds. Now consider that $\tx$ is never finalised and that it never
  becomes invalid with respect to the input of replica $r$ at any consensus instance posterior to
  $i$. Since a transaction is only cleared when it is included in the decided block of some
  consensus instance (Lin.\ref{lin:frcbc-clear} of \ref{alg:frc-based-blockchain-replica}) then
  $\tx$ will never be cleared from $r$'s pool. By Lin.\ref{lin:frc1bc-pick-fifo}, $\tx$ will be
  included in $r$'s input for every consensus instance $i'\geq i$, which satisfies (iii) of
  \emph{Stubborn input} and we are done.
\end{proof}

% All these properties are sufficient for BC-FRC to uphold \emph{User fairness}.

Our main contribution is collected in Thm.\ref{thm:bcfrc-fairness} below. 

\begin{theorem}
  \label{thm:bcfrc-fairness}

  BC-FRC enjoys User fairness.
\end{theorem}
\begin{proof}
    Consider a correct implementation of BC-FRC. Assume a correct client issues a transaction $\tx$
    that is valid with respect to some correct replica $r$. If $\tx$ ever becomes invalid with
    respect to $r$ afterwards, then we are done. Otherwise, we consider the case where $\tx$ never
    becomes invalid with respect to $r$, and show that $\tx$ will be eventually finalised.
    
    Now, assume towards a contradiction that $\tx$ is never finalised. By \emph{Request agreement}
    each correct replica $r$ eventually receives the request to insert $\tx$ at some consensus
    instance $i_r$. Let $i$ be the biggest among the $i_r$’s. By \emph{Stubborn input}, there
    exists a consensus instance $i_0\geq i$ such that for each correct replica $r$ one of the three
    conditions of the \emph{Stubborn input} property of \S\ref{sec:BC-FRC} is true. The (i) and
    (ii) cannot be true for any correct replica, because it is straightforward to derive a
    contradiction to our assumptions from them. Therefore, \emph{Stubborn input} guarantees that
    (iii) will be true for every correct replica, this is, each correct replica $r$ inputs a block
    that contains $\tx$ at every consensus instance $i'\geq i_0$. By \emph{Fusion validity}, the
    decided block $\seqAt{b}{i_0}$ at consensus instance $i_0$ is contributed by $f+1$ or more
    replicas, and therefore there is at least one correct replica who contributed to
    $\seqAt{b}{i_0}$ whose sub-proposal contains $\tx$. By \emph{Fair fusion}, transaction $\tx$ is
    in the sequence of the decided block $\seqAt{b}{i_0}$ or otherwise $\tx$ would have become
    invalid, which we assumed it never happens. This means that $\tx$ is output at consensus
    instance $i_0$, and by \emph{Chain finality}, the transaction $\tx$ will never be revoked,
    hence it is finalised, which results in a contradiction and we are done. %
\end{proof}

In the next subsection we show that DL-SMR and BC-FRC are equivalent, since they reduce to each other.

\subsection{Equivalence of DL-SMR and BC-FRC}

% \begin{figure*}
%   \begin{center}
%       \begin{tikzpicture}{\small
%           \node[abstraction] at (0,-1.2) {
%             \parbox{6.9cm}{
%               DL-SMR ($\textbf{in:}~\client.\submit(\tx) \mid \textbf{out:}~\seq{\ell}$)\\
%               {Safety, User fairness, Log validity, Log finality.}
%             }
%           } ;
%           \node at (8,-1.75) {{\large $\equiv$}} ;
%           \node[abstraction] at (8.8,0) {
%             \parbox{6.9cm}{
%               BC-FRC ($\textbf{in:}~\client.\request(\tx) \mid \textbf{out:}~\seq{b}$)\\
%               {Valid request, Valid input, Input agreement, Stubborn input.}
%               \\ \\ \\ \\ \\ \\ [9pt]
%             }
%           } ;
%           \node[abstraction] at (8.9,-1.1) {
%             \parbox{6.7cm}{
%               FRC ($\textbf{in:}~\replica.\inputBlock(b) \mid \textbf{out:}~\seq{b}$)\\
%               {Fair fusion, Fusion Validity.}
%               \\ \\ \\ [5.5pt]
%             }
%           } ;
%           \node[abstraction] at (9,-2.2) {
%             \parbox{6.5cm}{
%               RC ($\textbf{in:}~\replica.\inputBlock(b) \mid \textbf{out:}~\seq{b}$)\\
%               {Agreement, Termination, Chain validity, Chain finality.}
%             }
%           } ; }
%       \end{tikzpicture}
%   \end{center}
%   \caption{Relation among DL-SMR, BC-FRC, FRC and RC (symbol $\equiv$ indicates that
%     DL-SMR and BC-FCR reduce to each other).}
%   \label{fig:relation-dlsmr-bcfrc-frc-rc}
% \end{figure*}

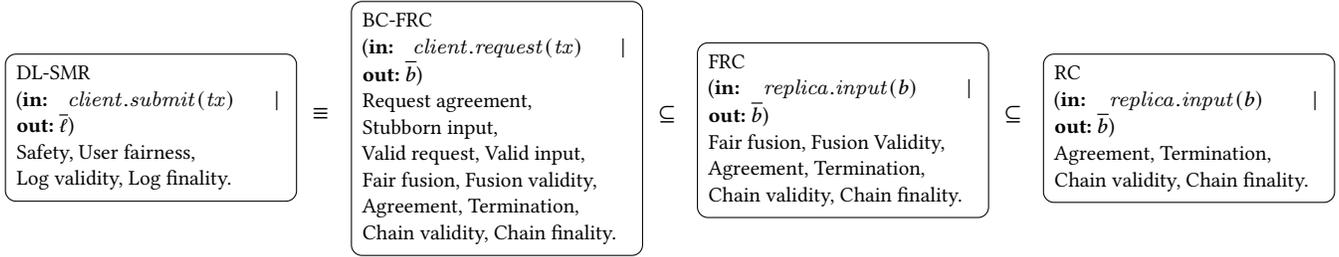
\begin{figure*}
  \begin{center}
      \begin{tikzpicture}{\small
          \node[abstraction] at (0,0.8) {
            \parbox{3.5cm}{
              DL-SMR\\
              ($\textbf{in:}~\client.\submit(\tx) \mid \textbf{out:}~\seq{\ell}$)\\
              {Safety, User fairness,\\Log validity, Log finality.}
            }
          } ;
          \node at (4.2,0) {{\large $\equiv$}} ;
          \node[abstraction] at (4.6,1.5) {
            \parbox{3.5cm}{
              BC-FRC\\
              ($\textbf{in:}~\client.\request(\tx) \mid \textbf{out:}~\seq{b}$)\\
              {Request agreement,\\Stubborn input,\\Valid request, Valid input,\\
               Fair fusion, Fusion validity,\\
               Agreement, Termination,\\Chain validity, Chain finality.}
            }
          } ;
          \node at (8.8,0) {{\large $\subseteq$}} ;
          \node[abstraction] at (9.2,.95) {
            \parbox{3.5cm}{
              FRC\\
              ($\textbf{in:}~\replica.\inputBlock(b) \mid \textbf{out:}~\seq{b}$)\\
              {Fair fusion, Fusion Validity,\\Agreement, Termination,\\Chain validity, Chain  finality.}
           }
          } ;
          \node at (13.4,0) {{\large $\subseteq$}} ;
          \node[abstraction] at (13.8,0.8) {
            \parbox{3.5cm}{
              RC\\
              ($\textbf{in:}~\replica.\inputBlock(b) \mid \textbf{out:}~\seq{b}$)\\
              {Agreement, Termination,\\Chain validity, Chain finality.}
            }
          } ; }
      \end{tikzpicture}
  \end{center}
  \caption{Relation among DL-SMR, BC-FRC, FRC and RC.}
  \label{fig:relation-dlsmr-bcfrc-frc-rc}
\end{figure*}

The fairness guarantees provided by BC-FRC warrant the reduction from DL-SMR to BC-FRC.

\begin{lemma}
  \label{lem:dlsmr-reducible-bcfrc}

  DL-SMR is reducible to BC-FRC.
\end{lemma}
\begin{proof}
  Assume a correct implementation of BC-FRC. We produce a correct implementation of DL-SMR that
  uses the implementation of BC-FRC as a black box and that does not rely on partially synchronous
  communication as follows. Upon the invocation of $\submit(\tx)$ by a correct client, we let the
  client trigger the primitive $\request(\tx)$ of BC-FRC. We let the predicate $\validTxs$
  coincide with $\validBlocks$ when the chain produced by BC-FRC is viewed as a log by
  concatenating the sequences of transactions contained in the blocks in the order they
  occur in the chain.

  \emph{Safety} holds by \emph{Agreement} when considering that a replica commits a transaction
  $\tx$ when it outputs a block whose sequence of transactions contains $\tx$. (Without loss of
  generality, we let the replica commit at once all the consecutive log positions that correspond
  to appending to the current log the sequence of transactions in the output block.) \emph{User
  fairness} holds by Thm.\ref{thm:bcfrc-fairness} and by our choice of $\validTxs$ and the
  correspondence between chains and logs. \emph{Log validity} holds by \emph{Chain validity} and by
  our choice of $\validTxs$, and \emph{Log finality} hods by \emph{Chain finality}.
\end{proof}

The reduction can be established in the other direction too.

\begin{lemma}
  \label{lem:bcfrc-reducible-dlsmr}

  BC-FRC is reducible to DL-SMR.
\end{lemma}
\begin{proof}
  Assume a correct implementation of DL-SMR. We produce a correct implementation of BC-FRC that
  uses the implementation of DL-SMR as a black box and that does not rely on partially synchronous
  communication as follows. Upon the invocation of $\request(\tx)$ by a correct client, we let the
  client trigger the primitive $\submit(\tx)$ of DL-SMR. We let the predicate $\validTxs$ coincide
  with $\validBlocks$ when we consider that a chain corresponds to the log that results from
  appending the sequences of transactions in the blocks in the same order.
  
  \emph{User fairness} holds by Thm.\ref{thm:bcfrc-fairness}. The properties of \emph{Safety},
  \emph{Log validity} and \emph{Log finality} hold respectively by \emph{Agreement}, \emph{Chain
  validity} and \emph{Chain finality} and by our choice of $\validTxs$.
\end{proof}

Lem.\ref{lem:dlsmr-reducible-bcfrc}--\ref{lem:bcfrc-reducible-dlsmr} above help to show that DL-SMR
and BC-FRC solve essentially the same problem.

\begin{theorem}
  \label{thm:smr-bcfrc-equivalent}

  DL-SMR and BC-FRC are equivalent.
\end{theorem}
\begin{proof}
  Direct consequence of Lem.\ref{lem:dlsmr-reducible-bcfrc}--\ref{lem:bcfrc-reducible-dlsmr}.
\end{proof}

Figure~\ref{fig:relation-dlsmr-bcfrc-frc-rc} summarises the content of
Lem.\ref{lem:frc-reducible-rc}--\ref{lem:bcfrc-reducible-frc},
Lem.\ref{lem:dlsmr-reducible-bcfrc}--\ref{lem:bcfrc-reducible-dlsmr},
Thm.\ref{thm:bcfrc-fairness}, and Thm.\ref{thm:smr-bcfrc-equivalent}

Our solution relies on five conditions that encompass the issuing
of transactions by clients, the collection of these transactions into blocks by replicas,
the input of a block to the RC service at each replica, and the output of the decided
block at each replica. More precisely, our solution requires the following:
\begin{enumerate}[label=(\arabic*)]
\item All correct replicas eventually receive all the transactions issued by correct
  clients.
\item The correct replicas collect the transactions they receive in order, and they will
  stubbornly input blocks that contain these transactions in receive order, until each of
  the transactions is either finalised or has become invalid with respect to the current
  chain.
\item \label{cond:two} At each correct replica, the RC service may aggregate inputs coming
  from different contributors (\ie, replicas) into a single block. In the vein of vector
  consensus \cite{DS97}, a replica will only output a decided block if it has been contributed by
  $f+1$ or more replicas,\footnote{\emph{Vector validity} requires the decided vector to contain
  $2f+1$ or more non-null elements, but for our purposes $f+1$ contributors is enough.}  which
  guarantees that the decided block is contributed by at least one correct replica.
\item We generalise vector consensus by considering a deterministic, application-dependent
  function $\fusion$ that takes several inputs and aggregates them into a single block. We
  require the decided block in \ref{cond:two} above to coincide with the aggregated block
  that results from applying $\fusion$ to the inputs of $f+1$ or more contributors. We
  make no particular assumptions on the order in which $\fusion$ places each of the
  transactions in the aggregated block, but we require that $\fusion$ does not drop any
  transaction unless including it in the aggregated block would make the block invalid
  with respect to the current chain. This provides great flexibility since the $\fusion$
  function can capture different existing models for \emph{order-fairness}
  \cite{KMZGSJA20}. \AG{Add citations for order-fairness, from \cite{KMZGSJA20}.}
\item The transactions that are inserted into the blockchain are never revoked. This
  requirement is always true for the BC-RC constructions in the permissioned setting that
  we consider, since the RC service is assumed to ensure absolute finality.
\end{enumerate}
\AD{here the idea is to give the intuitions in a simple way, too check if there are things
  that can be made more informal and easy to grasp}

  \begin{comment}
  \AG{We need
    to fix a particular correct history of BC-FRC for each history of DL-SMR, in the
    spirit of the emulations in the reducibility notion in\cite{CT1996} and the
    simulations in \cite{FORY09}. The $\fusion$ function can be fixed to the trivial
    function by letting $f+1$ correct clients propose the same set of transactions in the
    BC-FRC history. Elaborate on our precise notion of reducibitlity by taking this into
    account.}
    \end{comment}
% \AD{Comment on how the sufficient conditions entail the necessary ones, in particular
% infinite contribution.}

Implementations of BC-RC abound \cite{Gramoli17,Sousa18,libraSMR, tendermint} and the
result stated in Thm.\ref{thm:smr-bcfrc-equivalent} above and the generic BC-FRC suggest
avenues for implementing DL-SMR in the blockchain setting. However, in order to apply this
result broadly, the first thing to do is to ensure that the consensus algorithm at the
core of existing blockchains implements FRC. To this end, the next section introduces
transformations that deliver a correct implementation of BC-FRC given an existing
implementation of BC-RC.

This section introduces a transformation to obtain BC-FRC from a correct implementation
of RC. %
% {\color{red} with the Validity property extended (as the RC validity in \cite{DDFPT08}):
% If $\seqPrUpTo{b}{r}{i}$ is the current output of a correct replica~$r$, then
% $\validBlocks(\seqPrUpTo{b}{r}{i})$ holds and if $b^j$ in $\seqPrUpTo{b}{r}{i}$ then
% $b^j$ is the input at some replica during the instance $j$. \footnote{Here we are making
% explicit what already happens in current implementations, e.g., in Tendermint
% \cite{tendermintv2} processes verifies that the proposed block comes from the expected
% proposer before verifying that it satisfied the external predicate.}}
More into details, this transformation turns a correct implementation of RC into a correct implemeantion
of FRC, by aggregating the sub-proposals that a contributor hears of into a single
proposal before inputing it to RC. The obtained implmentation of FRC could be used to
obtain BC-FRC in a way analogous to Alg.\ref{alg:frc-based-blockchain-replica} of
\S\ref{sec:relation-rc-smr}. %Although this transformation is straightforward and correct,

\section{Related work}\label{sec:related-work}

\subparagraph*{Blockchains and SMR} The SMR problem, proposed in \cite{lamport78,SMR} and
generalized in \cite{SMR}, gained interest since Castro and Liskov \cite{pbft} showed how
to practically implement it in an eventually synchronous system in presence of Byzantine
processes. Blockchains can be seen as a way to implement SMR, especially when they support
smart contracts \cite{Wood14}, \ie, pieces of executable code stored and executed in a
replicated way in the blockchain. Nonetheless, to the best of our knowledge, no current
market blockchain solution provides complete support for SMR. Libra \cite{libraSMR} is one
of the rare cases in which the blockchain is expressed as a construction over SMR, however
the proposed SMR specification is restricted and does not encompass user
fairness.\footnote{In Libra SMR is expressed in terms of safety---all honest nodes observe
  the same sequence of commands---and liveness---new commits are produced as long as valid
  commands are submitted.}  Interestingly, formalisations of known blockchains put
emphasis on the replicated data structure and its maintenance, but leave aside user
fairness \cite{Bt-adt,AntaKGN18,tenderbake,GarayKL15}.

\subparagraph*{Fairness in blockchains}
Fairness is generally desirable when participants of a blockchain have some stake into the
system. An example of fairness can be found in the mutual exclusion problem, where
fairness demands that all processes requiring to access the critical section must be
eventually served \cite{dijkstra1965,panettiere}. In the blockchain context the notion of
fairness has mainly the connotations of
\begin{itemize}
\item[---] \emph{miner fairness}, which requires that processes that maintain the
  blockchain are rewarded proportionally to their amount of work, and
\item[---] \emph{user fairness}, which requires that user requests are eventually served
  by the blockchain, in the spirit of fairness for mutual exclusion.
\end{itemize}
\AG{Consider removing the pragraph and saying taht we focus on user fairness? Do not
  mistake our user fairness with previous notion of miner fairness or chain quality (keep
  citations).} One of the first analysis in blockchain systems on miner fairness is the
study of \emph{chain quality}, which is defined by Garay \etal\ \cite{GarayKL15} as the
proportion of blocks mined by honest miners in any given window. Other works that explored
miner fairness are Ourobouros \etal\ \cite{KRDO17} and Fruitchain \cite{PS17}. Although
miner fairness can potentially affect user fairness, it is not a necessary condition to
achieve user fairness, and therefore these previous works on chain quality do not close
the gap that this paper focuses on.

Regarding user fairness, Gürcan \etal\ \cite{GRT18,GDT17} study Bitcoin and show that once
a user issues a transaction (a request to the blockchain) there is no guarantee that such
transaction will be committed to the blockchain. Herlihy and Moir \cite{HM16} focus on
blockchains based on BFT consensus---such as Tendermint and Libra---and discus how
malicious participants could violate user fairness by censoring transactions. Indeed, this
is a common problem in such blockchains because the decided block for a given consensus
instance could have been proposed by a Byzantine replica. To overcome this problem, in
\cite{HM16} they propose modifications to Tendermint's algorithm to make user-fairness
violations detectable and accountable.

To the best of our knowledge, Fairledger \cite{fairledger} is one of the first blockchains
based on synchronous BFT consensus that explicitly provides user fairness. User fairness
is guaranteed because at each consensus instance all the proposals are batched in a sole
input value, so that either all of the transactions issued by clients are committed, or
none of them are. In the same spirit, Honeybadger \cite{honeybadger} uses an asynchronous
component that implements randomized agreement, and batches together proposals from
different replicas before inputting the aggregated proposal to the randomized agreement
service. \AG{Does this sentence refer to an asynchronous component of Honeybadger that is
  fed with aggregated proposals? I failed to understand the original sentence.}  Moreover,
they provide a probabilistic bound on the delay within which a request is committed, after
the request has been received by enough correct processes.

Recently, in a parallel and independent work, Crain \etal\ \cite{CNG21} propose a solution similar to ours to achieve censorship-resistance in Red Belly. They aggregate the ``valid and non-conflicting transactions'' coming from $f+1$ proposers into a `superblock' that is later input to the consensus module of Red Belly. They also introduce an abstraction for \emph{Set Byzantine Consensus} which resembles our FRC, but their abstraction is specific for Red Belly's cryptocurrency application since their notion of validity captures only that the transactions in the decided block do not exhaust the balance of any account. Our \emph{Valid fusion} and \emph{Fusion validity} properties allow for any generic-purpose application while capturing different existing approaches for fair-orderness \cite{KMZGSJA20}. We introduce a definition of SMR with a rigorous \emph{User fairness} property and prove that our solution achieves SMR.

\subparagraph*{Techniques to achieve users fairness in BC-RC} The approaches in both
Fairledger and Honeybadger take inspiration from techniques emerged in the context of
reductions of SMR to consensus. \AG{In Fairledger the decision is atomic, either all or
  none.} For instance, Honeybadger is an improvement in terms of communication complexity
of the asynchronous atomic broadcast solution from Cachin \etal\ \cite{CKPS01}. In
\cite{CKPS01} they show how to reduce atomic broadcast to consensus provided that the
latter enjoys an external validity property that ensures that the decided value satisfies
a validity predicate that can be verified locally by the replicas. Such an external
valdity property opens up for the possibility that a value complying with the validity
predicate is decided, even if proposed by a Byzantine replica. This situation is customary
in the blockchain setting, where a valid block (with respect to a local predicate) can be
commited to the chain even if produced by a Byzantine process. The relevance of their work
comes also from the fairness property provided by the atomic broadcast abstraction, which
ensures that a payload message $m$ is scheduled and delivered within a reasonable
(polynomial) number of steps after it is atomically broadcast by an honest process. The
techniques they use are in the spirit of Doudou and Shiper's solutions for certified
consensus and for atomic broadcast \cite{DS97}. \AG{What about Fairledger?}

Milosevic \etal\ provide a thorough survey of reductions of atomic broadcast to RC in the
Byzantine model, and analyse the validity property of different variants of RC to show
which of them are equivalent to atomic broadcast \cite{Milosevic11}. \AG{Explore relation
  between external and internal validity.}

\AG{A new paragraph for the blockchains?}

Sousa and Bessani \cite{SB12} provide an interesting solution of SMR in terms of an
implementation of RC that, despite Byzantine replicas, guarantees that all operations from
correct clients are always executed. To do so, they adapt the PBFT protocol \cite{pbft,
  pbftsmart} in such a way that a leader change is triggered when a client transaction is
not committed after a certain timeout, a fact that hints at the leader possibly being
Byzantine. Note that Sousa and Bessani's solution is not generic in the sense that it
depends on the specific implementation of RC that they consider, while our generic
solution applies to any implementation of RC and therefore it could be applied to any
existing blockchain.%\footnote{The only requirement for our solution is that the
  %implementation of RC carries out a round of voting in order to decide the value of each
  %consensus instance, a requirement which is arguably met by every implementation of RC
  %that uses standard BFT consensus protocols.} 
  \AG{Rephrase for precision and clarity, use
  the following ``The pipelining optimisation takes a BC-RC as a white box, since it
  modifies the implementation of the underlying RC, and delivers a DL-SMR. This reduction
  is agnostic to the consensus protocol, while Sousa and Bessani's is not''.} In a similar vein, Prime \cite{prime} guarantees that once a correct replica gets to know about a transaction
  then such transaction is committed in a timely manner. Do to so, a Pre-Ordering phase precedes the 
  PBFT consensus in charge to perform a global ordering on the transactions. Still in the
SMR context, Bazzi and Herlihy propose Clairvoyant \cite{Clairvoyant}, a new solution to
the problem in the form of a protocol for the generalized consensus problem
\cite{Lamport05}. \AG{Improve precision of this sentence.}  Interestingly, it provides
liveness for all transactions under eventual synchrony and at all times for transactions
that do not overlap with conflicting transactions. \AG{Conflicting transactions is a
  slippery concept, we must define it formally before stating anything of this kind.} The
solution leverages on non-skipping time stamps, which requires clients to interact with
replicas to get the latest time stamp they observed before issuing a transaction. Their
solution assumes $4f+1$ or more replicas, where $f$ is the failure threshold. %
Close to us, there is a very recent work \cite{zhang2020} that investigates how to avoid both transaction censorship and reordering despite the presence of Byzantine leaders proposing a block. To do so, when adding each block, two extra phases precede the consensus instance. Which, interestingly, can be any leader based consensus solution. Contrarily to us, \cite{zhang2020}, provides ordering but at the price of extra phases for each block creation. In our case we provide censorship resistance without transaction ordering guarantees adding only one extra phase for any leader based consensus solution. \AD{to recheck} 
\AG{Use the necessary conditions to show that popular blockchains do not meet user
  fairness. Consider moving related work to the end and merge it with conclusions and
  future work.}

\AG{Perhaps mention that we provide (definite) User fairness and say that the
  probabilistic user fairness porperty could be explored too?}

%%% Local Variables:
%%% fill-column: 90
%%% require-final-newline: t
%%% mode-require-final-newline: t
%%% next-line-add-newlines: nil
%%% show-trailing-whitespaces: t
%%% indent-tabs-mode: nil
%%% ispell-dictionary: "british"
%%% mode: latex
%%% TeX-PDF-mode: t
%%% TeX-master: "main"
%%% End:

\section{Conclusion}\label{sec:conclusion}

In this work we analysed how to achieve state machine replication in blockchains based on repeated consensus. Our work systematises the huge body of work on state machine replication  and its relationship with well-known distributed system problems such as Consensus and Atomic Broadcast. Our systematisation proposes formal abstractions to capture state machine replication and consensus implemented in blockchains, i.e. DL-SMR and BC-RC.  We pointed out that what is missing in BC-RC to achieve DL-SMR is the user fairness property, stipulating that transactions issued by correct participants are eventually committed, if valid. 
We further introduced  the FRC abstraction and a parametric BC-FRC construction, which enjoys
user fairness, and state formally the equivalence between DL-SMR and BC-FRC.  Notably our FRC abstraction relies on a fusion function and a validity predicate that  could be instantiated to cover a varied spectrum of order-fairness models. 
Finally, we
showed how to obtain DL-SMR out of RC in an agnostic way to the underlying consensus
protocol and then truly modular. 

%%% Local Variables:
%%% fill-column: 90
%%% require-final-newline: t
%%% mode-require-final-newline: t
%%% next-line-add-newlines: nil
%%% show-trailing-whitespaces: t
%%% indent-tabs-mode: nil
%%% ispell-dictionary: "british"
%%% mode: latex
%%% TeX-PDF-mode: t
%%% TeX-master: "main"
%%% End:

% \begin{acks}
% To \ldots
% \end{acks}

\bibliographystyle{ACM-Reference-Format}
\bibliography{main}


\end{document}